\documentclass[twoside,leqno,twocolumn]{article}
\usepackage{amsthm}
\newtheorem{theorem}{Theorem}[section]
\newtheorem{lemma}{Lemma}[section]

\newtheorem{proposition}{Proposition}[section]

\newtheorem{Definition}{Definition}[section]

\usepackage{framed}
\usepackage{enumerate}
\usepackage{xspace}
\usepackage{color,graphicx}
\newcommand{\set}[1]{\{#1\}}
\usepackage{amsmath,amssymb}
\usepackage[small,compact]{titlesec}
 
 \usepackage[small,it]{caption}

 \newcommand{\squishlist}{
 \begin{list}{$\bullet$}
  { \setlength{\itemsep}{0pt}
     \setlength{\parsep}{3pt}
     \setlength{\topsep}{3pt}
     \setlength{\partopsep}{0pt}
     \setlength{\leftmargin}{1.5em}
     \setlength{\labelwidth}{1em}
     \setlength{\labelsep}{0.5em} } }
\newcommand{\squishlisttwo}{
 \begin{list}{$\bullet$}
  { \setlength{\itemsep}{0pt}
     \setlength{\parsep}{0pt}
    \setlength{\topsep}{0pt}
    \setlength{\partopsep}{0pt}
    \setlength{\leftmargin}{2em}
    \setlength{\labelwidth}{1.5em}
    \setlength{\labelsep}{0.5em} } }
\newcommand{\squishend}{
  \end{list}  }
\newcounter{Lcount}
\newcommand{\squishenum}{
\begin{list}{\arabic{Lcount}. }
{ \usecounter{Lcount}
\setlength{\itemsep}{0pt}
\setlength{\parsep}{0pt}
\setlength{\topsep}{0pt}
\setlength{\partopsep}{0pt}
\setlength{\leftmargin}{2em}
\setlength{\labelwidth}{1.5em}
\setlength{\labelsep}{0.5em} } }

\DeclareMathOperator{\act}{active}

\DeclareMathOperator{\dist}{dist}

\DeclareMathOperator{\len}{length}

\DeclareMathOperator{\OPT}{OPT}
\newcommand{\piall}{\pi^{\text{\textup{all}}}}
\newcommand{\piin}{\pi^{\text{\textup{in}}}}
\newcommand{\piout}{\pi^{\text{\textup{out}}}}
\usepackage{hyperref}
\begin{document}
\title{An efficient polynomial-time approximation scheme for Steiner forest in planar
  graphs} 

\author{David Eisenstat\thanks{Department of Computer Science, Brown University.  Supported in part by NSF grant CCF-0964037.}
        \and 
        Philip Klein\footnotemark[1] \and Claire Mathieu\footnotemark[1]
}
\date{}
\maketitle 
\begin{abstract} 
We give an $O(n \log^3 n)$ approximation scheme for Steiner forest in
planar graphs, improving on the previous approximation scheme for this
problem, which runs in $O(n^{f(\epsilon)})$ time.
\end{abstract}
\setcounter{page}{0}
\newpage
\section{Introduction}

In the {\em Steiner forest} problem, we are given an undirected graph $G$ with
edge-lengths and a set $\mathcal D$ of pairs $(s_i,t_i)$ of vertices.  The
pairs are called {\em demands}, and the vertices that appear in demands are
called {\em terminals}.  The goal is to find a minimum-length forest $F$  that,
for every demand $(s_i, t_i)$, contains a path in $F$ from $s_i$ to $t_i$.
This problem generalizes the Steiner tree problem in networks.

There is a polynomial-time 2-approximation algorithm~\cite{AKR95}, but the
problem doesn't have an approximation scheme unless P=NP~\cite{BP89,Thimm01}.
However, for restricted input classes, polynomial-time approximation schemes
have been found.  For the case where the vertices are the points on the plane
and edge-lengths are Euclidean distances, Borradaile, Klein, and
Mathieu~\cite{BKM08} give an approximation scheme that can be implemented in
$O(n \log n)$ time where $n$ is the number of terminals.

For planar graphs, Bateni, Hajiaghayi, and Marx~\cite{BateniHM10} give a
polynomial-time approximation scheme.  The running time, however, for obtaining
a $(1+\epsilon)$-approximate solution has the form $n^{\epsilon^{-c}}$.  The
degree of the polynomial grows as $\epsilon$ gets smaller.  An {\em efficient}
polynomial-time approximation scheme is an approximation scheme whose running
time has the form $O(f(\epsilon) n^c)$ for some function $f$ and some constant
$c$ independent of $\epsilon$.  Thus the approximation scheme of Bateni,
Hajiaghayi, and Marx is not an efficient PTAS in this sense.  Our main result
is an efficient PTAS:

\begin{theorem} \label{thm:main}
For planar Steiner forest, there is an approximation scheme whose running time
is $O(n \log^3 n)$.
\end{theorem}

\section{Techniques}

\subsection{Branchwidth}

\begin{figure*}
\centerline{\includegraphics{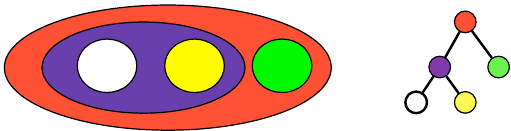}\hfill\includegraphics{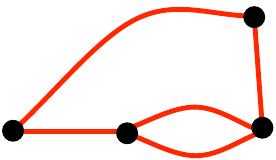}\includegraphics{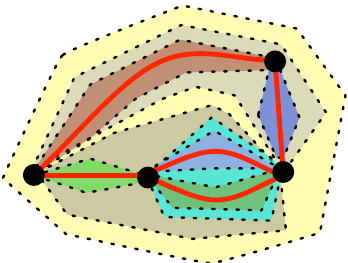}}
\caption{Left: a carving. Right: a branch decomposition.}
\end{figure*}

Tree-decomposition and branch-decomposition are ways to map pieces of the graph
to nodes of a tree so that, loosely speaking, pieces have small overlap.   We
formally define branch-decomposition.  The general paradigm is to reduce the
problem to graphs of bounded tree- or branch-width.

A {\em carving} of a ground set is a maximal family $\mathcal C$ of mutually
noncrossing subsets of the ground set.  In this paper, we refer to the sets in
$\mathcal C$ as {\em clusters}.  The Hasse diagram of the inclusion partial
order on the clusters is a tree in which each node has zero or two children.

A {\em branch-decomposition} of a graph $G$ is a carving $\mathcal C$ of the
edges of $G$.  The {\em boundary} $\partial C$ of a cluster $C$ is the set of
vertices $v$ such that $C$ contains a proper nonempty subset of the edges
incident to $v$.  The {\em width} of a branch-decomposition $\mathcal C$ of $G$
is $\max \set{|\partial C|\ :\ C \in \mathcal C}$.  The {\em branchwidth} of
$G$ is the minimum width over all branch-decompositions of $G$.

Treewidth (not defined here) is within a constant factor of branchwidth.
Graphs of bounded treewidth and branchwidth are tree-like, and many problems
can be solved exactly in linear time on such graphs.  As we will see, this is
{\em not} the case for Steiner forest.

\subsection{Framework}

The approximation scheme of Bateni, Hajiaghayi, and Marx fits into the
framework of Klein~\cite{Klein05b}, which consists of the following steps:
\begin{description}
\item[spanner] Find a subgraph $G_1$ (called the Steiner tree spanner) of the
input graph $G_0$ such that, for constants $c$ and $d$,
\begin{enumerate}
\item $\len(G_1) \leq c \OPT(G_0))$, where $\OPT(G_0)$ is the total length of
the edges used by an optimum solution, and \label{light-prop}
\item $\OPT(G_1) \leq (1+d\epsilon) \OPT(G_0)$.  \label{spanning-prop}
\end{enumerate}
\item[thinning] Partition the edges into $p$ subsets such that the contraction
of any subset yields a graph of branchwidth $O(p)$.  Let $G_2$ be the graph
obtained from $G_1$ by contracting the subset $S$ having the smallest total
length.
\item[dynamic programming] Find an (approximately) optimal solution in $G_2$.
\item[lifting] Lift the solution in $G_2$ to a solution in $G_0$ by
uncontracting edges of $S$ and adding them to the solution as needed.
\end{description}
This presentation of the framework differs from the original in~\cite{Klein05b}
in that, in the original, the dynamic-programming step finds an optimal
solution.

Suppose the solution in the dynamic-programming step has length at most
$(1+c'\epsilon) \OPT(G_2)$.  Since contraction preserves connectivity,
$\OPT(G_2)\le \OPT(G_1)$. By the spanner property, $\OPT(G_1)\leq (1+d\epsilon)
\OPT(G_0)$.  The addition of some edges from $S$ in the lifting step increases
the length by at most $ \len(G_1)/p$.  We choose $p=c/\epsilon$ so the
additional length is at most $\epsilon \OPT(G_0)$.  Hence the length of the
final solution is $((1+c'\epsilon)(1+d\epsilon)+\epsilon) \OPT(G_0)$.

The dynamic-programming step is straightforward and takes linear time; the
construction is given in~\cite{Klein05b}.  (See also~\cite{DHM07,DHK05}.  (It
has been generalized to bounded-genus graphs \cite{DHM07} and, more recently,
to minor-excluded graphs~\cite{DHK05}.) The lifting step is problem-dependent
but  straightforward for the problems (such as TSP, Steiner tree, and Steiner
forest) to which the framework has been successfully applied.  The spanner and
dynamic-programming steps are problem-dependent.  It is in those steps that
Bateni, Hajiaghayi, and Marx~\cite{BateniHM10} introduced new techniques, and
it is there that our improvements go.

\subsection{Spanner}

The spanner step of Bateni et al.\ built on the same step in the Steiner tree
PTAS of Borradaile, Klein, and Mathieu~\cite{BKK07,BorradaileKleinMathieu2009}.
The proof of the latter can be adapted to show the following.  For an instance
of Steiner forest consisting of a graph $G$ and a set $\mathcal D$ of terminal
pairs, let $\OPT(G, \mathcal D)$ denote the optimum value.

\begin{lemma}[Borradaile et al., adapted] \label{lem:BKM-spanner}
For a number $\epsilon>0$,  a planar graph $G_0$, and a tree $T$ of $G_0$,
there is a subgraph $H$ of length $f(\epsilon) \len(T)$ such that, for any set
$\mathcal D$ of pairs of terminals belonging to $T$, $\OPT(H, \mathcal D) \leq
\OPT(G_0, \mathcal D)+\epsilon \len(T)$, where $f(\cdot)$ is a  fixed function.
Furthermore, $H$ can be constructed in $O(n \log n)$ time for fixed $\epsilon$.
\end{lemma}

To use this result, Bateni et al.\ introduced an algorithm called {\em
prize-collecting (PC) clustering}.

\begin{theorem}[Bateni et al.]
There is a polynomial-time algorithm that, given a number $\epsilon>0$ and a
(not necessarily planar) Steiner-forest instance $(G, \mathcal D)$, outputs a
partition ${\cal D}_1 \cup \dots \cup {\cal D}_\ell$ of $\cal D$ and
corresponding trees  $T_1, \ldots, T_\ell$ such that
\begin{enumerate}
\item the terminals comprising ${\cal D}_i$ belong to the tree $T_i$,
\item $\sum_{i=1}^\ell \len(T_i) \leq (\frac{4}{\epsilon}+2) \OPT(G, {\cal
D})$, and
\item $\sum_{i=1}^\ell \OPT(G, {\cal D}_i) \leq (1+\epsilon) \OPT(G, {\cal
D})$.
\end{enumerate}
\end{theorem}
To obtain a spanner for the Steiner forest instance $(G_{in}, \mathcal
D_{in})$, therefore, one can perform PC clustering, and then, for each tree
$T_i$, apply Lemma~\ref{lem:BKM-spanner} to obtain a spanner $H_i$.  The third
property of PC clustering implies that the union $\bigcup_{i=1}^\ell H_i$ will
be a Steiner-forest spanner for the original instance.

Gassner~\cite{Gassner10} showed that Steiner forest is NP-hard even in graphs
of treewidth~3.  Bateni et al.\ addressed this difficulty by giving an
(inefficient) PTAS for Steiner forest in bounded-treewidth graphs, one that
takes $n^{O(w^2/\epsilon)}$ time ($w$=width).

\subsection{Our improvements to the spanner step}

When we try to obtain a quasi-linear approximation scheme, the PC-clustering
algorithm of Bateni et al.\ fails us in two ways.
\begin{itemize}
\item The running time is $O(n^2 \log n)$.  Indeed, the running time is given
in~\cite{BateniHM10} as ``polynomial''; beyond that, it does not matter since
the overall time for their approximation scheme is $O(n^{f(\epsilon)})$.
\item Once the trees $T_1, \ldots, T_k$ are found, a spanner $H_i$ needs to be
found for each tree.  Finding a spanner, given $T_i$, takes $O(n \log n)$ time,
so the overall time for finding the spanners is $O(k n \log n)$.  Since $k$ is
$\Omega(n)$ in the worst case, the bound is $O(n^2 \log n)$.
\end{itemize}
We give a PC-clustering theorem that addresses both issues: the algorithm runs
in $O(n \log n)$ for planar graphs (in fact, for any excluded-minor family) and
it returns subgraphs $G_1, \ldots, G_\ell$ with small overlap (each edge is in
$O(\log n)$ subgraphs) in which the spanners can be
found.\footnote{PC-clustering can be stated in a somewhat more general way and
is used in this way in multiterminal cut; our result actually addresses the
more general problem.}

\begin{theorem}[New PC-clustering] \label{thm:PC-clustering}
For any $\delta>0$, there is an algorithm that, given $\epsilon>0$ and a
Steiner forest instance $(G, \mathcal D)$, outputs a partition ${\mathcal D}_1
\cup \dots \cup {\mathcal D}_\ell$ of $\mathcal D$ and corresponding trees
$T_1, \ldots, T_\ell$ and subgraphs $G_1, \ldots, G_\ell$ such that
\begin{enumerate}
\item the terminals comprising ${\cal D}_i$ belong to the tree $T_i$,
\item $\sum_{i=1}^\ell \len(T_i) \leq (\frac{4+\delta}{\epsilon}+2) \OPT(G,
{\cal D})$,
\item $\sum_{i=1}^\ell \OPT(G_i, {\cal D}_i) \leq (1+\epsilon) \OPT(G, {\cal
D})$,
\item each edge of $G$ is in $O(\log n)$ of the subgraphs.
\end{enumerate}
If the input graph $G$ is simple and planar or, more generally, comes from a
fixed excluded-minor family, the running time of the algorithm is $O(n \log
n)$.
\end{theorem}
Combining this algorithm with the $O(n \log n)$ construction of
Lemma~\ref{lem:BKM-spanner}, we obtain an $O(n \log^2 n)$ algorithm for
obtaining a Steiner-forest spanner for simple planar graphs.

PC-clustering is but one example of the use of primal-dual approximation
algorithms in approximation schemes for planar, bounded-genus, and
minor-excluded graphs.  Our technique for speeding up PC-clustering in planar
and bounded-genus graphs in such graphs applies to other primal-dual
approximation algorithms as well.  For example, the technique can be used on
the Goemans-Williamson approximation algorithm for prize-collecting TSP and
prize-collecting Steiner tree.  As a consequence, we obtain $O(n \log n)$
approximation schemes for these problems in planar and bounded-genus graphs.
The speed-up in the algorithm comes from use of a dynamic
data-structure~\cite{BrodalF99,Kowalik10} for maintaining orientations,
together with ideas from a data structure~\cite{Klein94} for efficient
implementation of primal-dual approximation algorithms.

\subsection{Our improvements to the dynamic programming step}

When we try to obtain an efficient approximation scheme, the dynamic program of
Bateni et al.\ fails us in one way: each tree that crosses the boundary of a
cluster is approximately represented by $O(w/\epsilon )$ of its vertices, and
there are $n^{O(w/\epsilon )}$ possible such vertex choices.

We take advantage of the spanner property in combination with the bounded
branchwidth property.  Recall that the graph $G_1$  has length at most $c
\OPT(G_0)$, and hence so does the graph $G_2$ resulting from the thinning step.
In Section~\ref{sec:algorithm}, we prove the following:
\begin{theorem} \label{thm:fast-dynamic-programming}
For any constant $\epsilon>0$, there is an $O(f(w) n \log^2 n)$ algorithm that,
for any instance $(G, \mathcal D)$ of Steiner forest of branchwidth $w$, finds
a solution of length at most $\OPT(G, \mathcal D)+\epsilon \len(G)$, where
$f(\cdot)$ is a fixed function.\footnote{There is no great significance to our
changing from treewidth to branchwidth.}
\end{theorem}

We achieve this using a new graph construction on branch-decompositions.  For
each cluster, if the sum of lengths of edges near the cluster's boundary is
high then the edges are contracted.  The result is a graph in which, for each
cluster, the sum of lengths of edges near the cluster's boundary is not too
big.  We can therefore cover the region near the boundary by a constant number
of regions of low diameter.  This simplifies the dynamic program since it
doesn't have to keep track of exactly where the terminals are---just which
regions contain them.  Since the number of regions is constant, we can get by
with fewer configurations.

The situation is a bit more complicated because the dynamic program
has to deal with regions at different scales, and has to guess the
scales.  We show it suffices to guess among a number of scales that is
logarithmic in the \emph{height} of the branch-decomposition tree and
exhibit a linear-time algorithm that, given an arbitrary branch
decomposition of width $w$, finds a new branch decomposition of width
$2 w$ and logarithmic height.  Another complication is the edges that
were contracted in the graph construction.  We show that, after
uncontracting these edges, the optimal solution can be patched so that
its length does not increase much.  Consequently, the solution found
by the dynamic program has length not much more than optimal.

Combining Theorem~\ref{thm:PC-clustering},
Theorem~\ref{thm:fast-dynamic-programming}, the $O(n \log n)$ spanner
construction of~\cite{BorradaileKleinMathieu2009} described in
Lemma~\ref{lem:BKM-spanner}, and the framework, we obtain
Theorem~\ref{thm:main}.

\section{Proof of Theorem~\ref{thm:PC-clustering}} \label{sec:PC-clustering}

\newcommand{\recentdead}{\text{RecentDead}}
\newcommand{\living}{\text{Living}}
\newcommand{\save}{\text{SAVE}}
\begin{figure*}\label{fig:pcclustering}
\begin{center}
\fbox{\parbox{6in}{\vspace{-2ex}\begin{tabbing}
{\small\bf PC-clustering, Phase 1:}\\
      \quad \= {\em input:} an initial graph $G$ with edge-lengths
      $\len(\cdot)$, and an initial assignment
      $\phi$ of budgets to vertices\\
      \> $t := 0$;\quad $\save := \emptyset$\\
      \> for each vertex $v$, \=$d[v] := 0$ \\
      \> while there is a living vertex\\
      \> \quad \=$\Delta_1 := \min \set{\phi[v]\ :\  v\in V(G), v \text{ living}}$\\
      \> \>$\Delta_2 := \min$\ \=$\set{\len(uv)\ :\ uv \in
        E(G), \text{ one of }\set{u,v}\text{ is living}} $\\
      \>\>\>${} \cup \set{\len(uv)/2\ :\ uv\in E(G), \text{ both $u$ and $v$ are living}}$\\
      \> \> $\Delta := \min \set{\Delta_1, \Delta_2}$ \# which happens first?\\
      \> \> $t := t+\Delta$ \# advance time\\
      \> \> for every living vertex $u$,\\
      \> \> \quad \= $d[u] := d[u] + \Delta$\\
      \>\>            \> $\phi[u] := \phi[u] - \Delta$ \\
      \> \>           \> $\len(uv) := \len(uv) - \Delta$ for every incident edge $uv$\\
      \> \> if some edge $uv$ now has zero length,\\
      \>\> \> contract $uv$, creating new vertex $w$\\
      \dag\> \>\> assign $\phi[w] := \phi[u] +\phi[v]$ and $d[w] := \max \set{d[u], d[v]}$\\
      \>\>\>if some endpoint (say $v$) is not living but $t <
      (1+\delta) d[v]$
      then add $uv$ to $\save$\\
      \>$F_1 := \set{\text{edges contracted}}$\\
{\small\bf PC-clustering, Phase 2:}\\
\> initialize $F_2 := F_1$\\
\> while there is an edge $e\in F_2-\save$ that is the only edge incident to a
dead vertex $v$\\
\> \quad delete $e$ from $F_2$
\end{tabbing}\vspace{-2ex}}}
\end{center}
\caption{The new PC-clustering algorithm}
\end{figure*}

We describe the algorithm that proves Theorem~\ref{thm:PC-clustering}.
It involves just a small change to the PC-clustering algorithm
of~\cite{BateniHM10}, although our presentation is different.

\subsection{Algorithm for basic PC-clustering with graph decomposition}
\label{sec:PC-clustering-algorithm}

In contrast to~\cite{BateniHM10}, we describe PC-clustering using contractions.
 When an edge $uv$ is
contracted, the endpoints are coalesced to form a new vertex.
The variable $t$ represents simulated
time.  Part of the input is an assignment $\phi[\cdot]$ of ``energy''
to vertices, that,  over time,consumed. 
When two endpoints $u,v$ of an edge are coalesced, the new vertex
combines their remaining energy.
We say a vertex$v$ is {\em living} if it has not yet exhausted its energy, i.e.
$\phi[v]>0$ (else {\em dead}), and $d[v]$ represents the amount of
(simulated) time $v$ has lived so far.  Our substantive change is to
introduce the notion of ``zombie'' vertices\cite{Romero}, which are
vertices that are joined to living vertices not too long (depending on
a parameter $\delta$) after they
die.  \\

\noindent Let $G_0$ denote the graph $G$ before the contractions of
Phase~1.  In the following, unless otherwise stated, the term {\em vertices}
includes the original vertices of $G_0$ as well as the new ones formed
by contraction.  For each vertex $v$, let $\phi_0[v]$ denote the
initial value of $\phi[v]$, the value when it is first assigned
(whether before Phase~1 commences, in the case of vertices of $G_0$,
or in line~\dag\ for vertices created by contractions).

The contractions define a binary forest, the {\em contraction forest},
on the vertices.  If an edge $uv$ was contracted and the resulting
vertex is $w$ then $u$ and $v$ are the two children of $w$ in the
contraction forest.

For each vertex $v$, let $S_v$ be the set of vertices of $G_0$ that were
coalesced to form $v$, and let $G_v$ be the subgraph of $G_0$ induced
by $S_v$.

We say an edge $e$ of $G_0$ is {\em incident} to a vertex $v$
if exactly one of the endpoints of $e$ in $G_0$ belongs to $S_v$.

A vertex $v$ is {\em isolated} if, at the end of Phase~2, no edge of
$F_2$ is incident to it.  We define the {\em isolated-dead-vertex
  forest} $\mathcal I$ to be the forest whose nodes are the isolated
dead vertices and such that the parent of $v$ is its nearest proper
isolated dead ancestor in the contraction forest.

\begin{lemma} \label{lem:connected-components}
For each vertex $v\in V(\mathcal I)$, there is a connected component
$T_v$ of $F_2$ whose vertex set is $S_v - \bigcup \set{S_w\
  :\ w \text{ a child of } v \text{ in } \mathcal I}$.
\end{lemma}


\begin{lemma} \label{lem:depth}
 The depth of $\mathcal I$ is at most $\displaystyle 1+\log_{1+\delta}
 \frac{\sum \set{\phi_0[u]\ :\ u\in V(G_0)}}{\min \set{\phi_0[v]\ :\ v\in
     V(G_0), \phi_0[v]>0}}$.
\end{lemma}

\begin{proof} At the end of Phase~1, for each dead vertex $v$, $d[v]$
  is the time when $v$ died.  Each root $r$ of
  $\mathcal I$ has $d[r] \leq \sum \set{\phi_0[u]\ :\ u\in V(G_0)}$.  
  Suppose $v$ is a nonroot vertex of $\mathcal I$, and let $w$ be the
  parent of $v$ in $\mathcal I$.  Let $uv$ be the edge contracted to
  form $w$.  Since $v$ is isolated, $uv$ does not remain in $F_2$, so
  $v\not\in \save$.  Therefore the time $t$ at which $uv$ is
  contracted must satisfy $t> (1+\delta) d[v]$.  Therefore $d[w] > (1+\delta)
  d[v]$.  This proves the lemma.
\end{proof}

\noindent The output of the algorithm is the forest $\mathcal I$ and,
for each vertex $v$ of $\mathcal I$, the subgraph $G_v$ and the
connected component $T_v$ of Lemma~\ref{lem:connected-components}.
Lemma~\ref{lem:depth} implies that each vertex/edge of $G_0$ is in a
logarithmic number of subgraphs.

\subsection{Length of forest returned by basic PC-clustering}

\begin{lemma} \label{lem:length}
The forest returned has length at most $2(1+\delta)\sum
\set{\phi_0[v]\ : \ v \in V(G_0)}$.  
\end{lemma}

\begin{proof}
For each value of $t$, let $\living(t)$ denote the set of vertices
that are living at time $t$, and let $\recentdead(t)$ denote the set of
vertices $v$ such that
$(\text{time of $v$'s death}) \leq t \leq (1+\delta)(\text{time of $v$'s death})$

For each vertex $v$, define $\tau[v] = \sum \set{\phi_0[u]\ :\ u \in
  S_v} - \phi[v]$.  Intuitively, $\tau[v]$ is the amount of energy
``used up'' by $v$ and its descendants in the contraction tree.
The algorithm ensures that $\phi[v]$ remains nonnegative, so $\tau[v]
\leq \sum \set{\phi_0[u]\ :\ u \in S_v}$.  Induction shows $d[v] \leq
\tau[v]$.  

In Phase~1, when an edge is added to $F_1$, its
its reduced length is zero.  The reduction in the length of edge
$uv$ can be attributed to the endpoint(s) living. 

For time $t$, let $G_t$ be the graph $G$ at time $t$, and let $H_t$ be
the edge-subgraph of $G_t$ consisting of edges that are in $F_2$ at
the end of Phase~2.  The total length of $F_2$ is at most
$\int \sum_{v\in\living(t)} (\text{degree of $v$ in $H_t$})\ dt$.
$H_t$ is a forest whose leaves are clusters that are
living or recent dead.  The degree in $H$ of dead clusters is at least
two, so $\sum_{C\in \living(t)} \deg_H(C) \leq 2(|\living(t)|
+|\recentdead(t)|)$.  Therefore the total length of $F_2$ is at most
$\int 2(|\living(t)| +|\recentdead(t)|)\ dt$.
While a vertex $v$ is living, $\phi[v]$ is decreasing at unit rate, so
$\tau[v]$ is increasing at unit rate.  This shows $\int |\living(t)|
dt \leq \sum \set{\phi[v] \ :\ v\in V(G_0)}$.
By the definition of $\recentdead(t)$ and 
$d[v] \leq \tau[v]$ shows
$\int |\recentdead(t)|)dt \leq \delta \sum \set{\phi[v] \ :\ v\in
  V(G_0)}$.  This proves the lemma.
\end{proof}

\paragraph{Remark} The only difference between this analysis and that
of~\cite{BateniHM10} is the part dealing with $\recentdead(t)$.

Recall that when two vertices coalesce, the resulting vertex gets the
remaining energy from its endpoints.  Therefore each bit of energy
possessed by a new vertex $v$ comes from some original vertex $u \in
S_v$.  Following~\cite{BateniHM10}, we think of the energy originally
assigned to $u$ as having the {\em color} $u$.  If some of $v$'s
energy comes from original vertex $u$, we will say that $v$ has color
$u$.  Let $E'$ be a set of edges of $G_0$.  We say a color $u$ is {\em
  exhausted} by $E'$ if every vertex $v$ colored by $u$ has an
incident edge in $E'$.  These concepts yield:

\begin{lemma}[Bateni et al.] \label{lem:exhaust-result} Let $L$ be the
  set of colors exhausted by $E'$.  The length of $E'$ is at least
  $\sum_{u\in L} \phi_u$.
\end{lemma}

\begin{lemma} \label{lem:exhaust-condition} An original vertex $u$ is
  exhausted by $E'$ if, for some dead vertex $v$ such that $u\in S_v$,
  $E'$ contains a path between $u$ and some original vertex not in
  $S_v$.
\end{lemma}

\subsection{Using PC-clustering in Steiner forest} \label{sec:Steiner-clustering}

Now we prove Theorem~\ref{thm:PC-clustering}.  (The proof of the
running time is in Section~\ref{sec:primal-dual}.)  
The input instance of Steiner forest is $(G_{in}, \mathcal D)$.
The
algorithm finds a 2-approximation solution $F^*$, and then obtains the
graph $G$ from $G_{in}$ by contracting
each connected component $K$ of $F^*$.  Let $k$ be the number of
components.  Let $Y$ be the set of components of length $<\frac{\epsilon}{2k}
\len(F^*)$.  For each component $K$ not in $Y$, the algorithm assigns
energy to the vertex $u$ of $G$ resulting from contracting $K$:
$\phi[u] := 2\epsilon^{-1}\ \len(K)$.  All other vertices of $G$
are assigned zero energy.

The algorithm runs Phase~1 and~2 of
Section~\ref{sec:PC-clustering-algorithm} on $G$ and $\phi[\cdot]$,
obtaining $F_2$ and the isolated-dead-vertex forest $\mathcal I$.
For each $v\in V(\mathcal I)$, the algorithm obtains a subgraph $G_v$ (see
Section~\ref{sec:PC-clustering-algorithm}) and (see
Lemma~\ref{lem:connected-components}) a connected component $T_v$ of
$F_2$.  For each, the algorithm obtains $G_v'$ from $G_v$ and $T_v'$ from $T_v$ by
uncontracting the edges of $F^*$, and defines $\mathcal D_v$ to be
the set of demands $(s, t)\in\mathcal D$ for which $s,t\in
V(T_v')$.  We claim that these structures satisfy the conditions in
Theorem~\ref{thm:PC-clustering}.  The first condition is satisfied by
construction.  Lemma~\ref{lem:length} implies that the second
condition is satisfied.  Since each vertex of $G$ is initially
assigned energy at least $\frac{1}{k}
\len(F^*)$, Lemma~\ref{lem:depth} implies that $\mathcal I$ has
depth $\leq 1+\log_{1+\delta} k$, which implies the fourth condition.

It remains to show the third condition, that the sum of optimum
values for the subinstances 
$\set{(G_v', \mathcal D_v)\ :\ v \in V(\mathcal I)}$
 is at most $1+\epsilon$ times the optimum value for
the original instance $(G_{in}, \mathcal D)$.  Since $\mathcal I$ is a
rooted forest, it induces a partial order on these subinstances.

Let $E'$ be the edge-set of an optimal solution to the original
instance.  Each connected component $C$ of $E'$ is assigned to the
subinstance 
 containing $C$ that is farthest
from a root in $\mathcal I$, i.e. the instance for which $G_v$ is smallest.
Let $H_v$ be the subgraph of $G_v'$ consisting of the components assigned to $(G_v',
\mathcal D_v)$.

Since $H_v$ might not constitute a feasible solution for that
instance, we might need to augment them.  Suppose that there is a
demand $(s, t)\in \mathcal D_v$ such that $s$ and $t$ are not
connected by $H_v$.  The 2-approximate solution $F^*$ contained some
connected component $K$ that joined $s$ and $t$; let $u_K$ be the vertex
in $G_0$ that resulted from contracting that connected component.  Since
$E'$ is a feasible solution, it too contained a connected component
that joined $s$ and $t$; since that component was not assigned to the
subinstance $(G_v', \mathcal D_v)$, it must be that the component is
not contained in $G_v'$, so, by Lemma~\ref{lem:exhaust-condition}, $u_K$
is exhausted by $E'$.  To augment $H_v$, we add the component $K$ of
the 2-approximate solution.  In the augmented solution, $s$ and $t$
are joined.  Either $K$ belongs to $Y$ or $\phi[u_K] =
2\epsilon^{-1}\,\len(K)$.  The sum of the lengths of components in
$Y$ is $\le k\cdot \frac{\epsilon}{2k} \len(F^*) \le
\frac{\epsilon}{2} \len(E')$, and $\sum
\set{\phi[u_K]\ :\ u_K \text{ is exhausted by } E'} \le \len(E')$, so
$\sum\set{\len(K)\ :\ K\not \in Y, u_K \text{ is exhausted by } E'}
\le \epsilon \len(E')$.  Therefore the sum of lengths of solutions to
the subinstances is $\le (1+\epsilon) \OPT(G_{in}, \mathcal D)$.  This
proves the third property of Theorem~\ref{thm:PC-clustering}.

\subsection{Primal-dual on planar and minor-excluded graphs} \label{sec:primal-dual}

We show that some primal-dual approximation algorithms, including
Goemans and Williamson's approximation algorithm for Steiner forest,
and Bateni, Hajiaghayi, and Marx's algorithm for PC clustering (and
our modification of this algorithm), can be implemented in $O(n \log
n)$ time for planar graphs.

The method is to combine an approach of~\cite{Klein94} to implementing
primal-dual approximation algorithms with a technique of~\cite{BrodalF99}

\subsubsection{Interface to data structure}

Klein~\cite{Klein94} shows that primal-dual algorithms such as
that of~\cite{GW95} can be implemented using a data structure.  There
are two categories ({\em active} and {\em
inactive} in the case of primal-dual).  An ordered pair $(c,c')$ of
categories is called a {\em bicategory}.  Each vertex $v$ is assigned
a to category $c(v)$, and thus
each edge $uv$ is assigned to a bicategory.
The data structure supports the following operations:
\begin{itemize}
\item
{\sc DecreaseCost}$(b, \delta)$, where $b$ is a bicategory and
$\delta$ is a real number, decreases by $\delta$ the cost of all
edges in bicategory~$b$. (2)
\item
{\sc FindMin}$(b)$ returns the minimum-cost edge in
bicategory~$b$.
\item
{\sc ChangeCategory}$(v,c)$ changes the category of $v$ to $c$
(implying changes to the bicategories of edges incident to $v$).
\item
{\sc ContractEdge}$(e,c)$ contracts $e$ and assigns the
resulting vertex to category~$c$.
\end{itemize}

\subsubsection{Representation}

Now we describe the data structure.
We use the ideas of~\cite{Klein94}
but make some changes to allow the data structure to be made more
efficient for graphs from a minor-excluded family.

The data structure maintains the following:
\begin{itemize}
\item
an orientation of the edges;
\item
an array $C[\cdot]$, indexed by vertices, such that $C[v]$ is
the category of $v$;
\item
an array $OUT[\cdot]$, indexed by a vertex $v$, such that
$OUT[v]$ is a linked list of the outgoing
edges of $v$;
\item an array $OUT[\cdot, \cdot]$, indexed by a vertex $v$ and a
category $c$, such that $OUT[v,c]$ is a linked list of the outgoing
edges of $v$ whose tails are in category~$c$;
\item
an array $IN[\cdot, \cdot]$, indexed by a
vertex $v$ and a category $c$, such that $IN[v, c]$ is a pointer to a
mergeable heap consisting of the incoming edges $uv$ of $v$ for which
the tail $u$ has category $c$;
\item
an array $B[\cdot]$, indexed by bicategories, such that $B[(c_1,
c_2)]$ is a heap consisting of $\set{IN(v,c_2)\ : C[v]=c_1}$.
\end{itemize}
Each edge and each heap has a real-number {\em label}.  The data
structure maintains the {\em label invariant}: the cost of an edge is the sum
of its label, the label of the heap $IN[v,b]$ that contains it, and the
label of the heap $B[b]$ that contains its heap.
The key of an edge in the priority queue that contains it is the
edge's label.  The key of a queue $IN[v,b]$ in the priority queue
$B[b]$ that contains it is the label of $IN[v,b]$ plus the minimum key
in $IN[v,b]$.

\subsubsection{Implementing {\sc DecreaseCost} and {\sc FindMin}}
{\sc DecreaseCost}$(b,\delta)$ is implemented by decreasing the label
of $IN(b)$ by $\delta$.  {\sc FindMin}$(b)$ is implemented by finding
the minimum heap in $B(b)$, and returning the minimum edge in that
heap.

\subsubsection{Implementing \sc ChangeCategory}$(v,c)$
Now we describe how to implement the operation {\sc
ChangeCategory}$(v,c)$.
Let $c_0 := C[v]$ (the old category of
$v$).  To handle the incoming
edges of $v$, for each category $c'$,  the heap $IN[v,c']$ is moved from $B[(c',
c_0)]$ to $B[(c', c)]$ (and the label of $IN[v,c']$ is adjusted to
preserve the label invariant).
To handle the outgoing edges of $v$, each
edge $vu$ in $OUT[v]$ is moved from $IN[u,c_0]$ to $IN[u,c]$ (and the
label of $vu$ is adjusted to preserve the label invariant).
Time required is $O(|\text{outgoing edges}| \log n)$.

\subsubsection{Implementing {\sc ContractEdge}}
To implement {\sc ContractEdge}$(uv,c)$, first delete the edge $uv$,
and change the categories of $u$ and $v$ to $c$.  Let $w$ denote the
vertex to be formed by coalescing $u$ and $v$.  For each category
$c'$, merge the heaps $IN[u, c']$ and $IN[v,
c']$, and assign the result to $IN[w, c']$; similarly, merge the lists
$OUT[u, c']$ and $OUT[v, c']$ and assign the result to $OUT[w, c']$.
Remove the edge $uv$ from the heap and list containing it.  Update the
tables to reflect the fact that $u$ and $v$ no longer exist.  The time
required is $O((|\set{\text{outgoing edges of $u$}}|+|\set{\text{outgoing edges
of $v$}}|) \log n)$.

\subsubsection{Maintaining bounded outdegree}

The time per operation is $O(\log n)$ if the outdegree of each vertex
is bounded.  Brodal and Fagerberg~\cite{BrodalF99} give a method for
dynamically maintaining bounded-outdegree orientations in families of graphs that
guarantee the existence of such orientations. Kowalik~\cite{Kowalik10}
points out that their method works with contractions.  Each update
takes amortized $O(\log n)$ time and changes the orientation of
$O(\log n)$ edges.

\section{Proof of Theorem~\ref{thm:fast-dynamic-programming}}
\label{sec:algorithm}

Now we prove Theorem~\ref{thm:fast-dynamic-programming}.  We are given
an instance $(G_{in},\mathcal D_{in})$ of Steiner forest, and a branch
decomposition of $G_{in}$ of width $w$.  For simplicity of
presentation, we want to assume that each edge has length~1.  To
justify this assumption, let $\eta= {\epsilon}{\len(G_{in})}/(cm)$,
and define a new length assignment $\widehat{\len}(e) := \lfloor
\len(e)/\eta \rfloor$.  Now all the lengths are integers, and the sum
of lengths is at most $c\epsilon^{-1} m$.  Replace edge $e$ with
$\widehat{\len}(e)$ edgelets (if $\widehat{\len}(e)=0$ then contract
$e$) to achieve the assumption.  Given a solution for the modified
instance, the additional length due to rounding is at most $\eta m$,
which by definition of $\eta$ is less than $\epsilon
\OPT(G_{in},\mathcal D_{in})$.

\subsection{Reducing the height of the branch decomposition}

\begin{figure*}
\begin{framed}
\parbox[b]{0.6\linewidth}{\begin{tabbing}
\qquad\=\qquad\=\qquad\=\kill
$\textsc{Balance}(\mathcal C)$:\\
\textbf{Input}: a branch decomposition $\mathcal C$ with root cluster $C$.\\
\textbf{Output}: a balanced branch decomposition $\mathcal C'$ with root cluster $C$.\\
Let $H_1 \subsetneq H_2 \subsetneq \cdots \subsetneq H_k = C$ be a heavy path, that is,\\
\> a maximal ascending chain such that, for all $i \in \{1, \ldots, k - 1\}$,\\
\> the sibling $H_{i + 1} - H_i$ of $H_i$ satisfies $\lvert H_{i + 1} - H_i \rvert \le \lvert H_i \rvert$.\\
Let $C_1' = H_1$.
For $i \in \{2, \ldots, k\}$, let $C_i' = H_i - H_{i - 1}$.\\
For $i \in \{1, \ldots, k\}$, let $\mathcal C_i' = \textsc{Balance}(\{D : D \in \mathcal C,\:D \subseteq C_i\})$.\\
Let $\mathcal C' = \textsc{Complete}(\mathcal C_1', \ldots, \mathcal C_k')$.
\end{tabbing}}
\hfill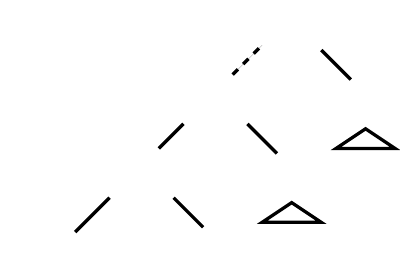

\noindent
\parbox[b]{0.6\linewidth}{\begin{tabbing}
\qquad\=\qquad\=\qquad\=\kill
$\textsc{Complete}(\mathcal C_1, \ldots, \mathcal C_k)$:\\
\textbf{Input}: branch decompositions $\mathcal C_i'$ with root clusters $C_i'$.\\
\> Clusters $C_1', \ldots, C_k'$ are pairwise disjoint.\\
\textbf{Output}: a branch decomposition $\mathcal C' \supseteq \mathcal C_1' \cup \cdots \cup \mathcal C_k'$\\
\> with root cluster $C_1' \cup \cdots \cup C_k'$.\\
For $i \in \{1, \ldots, k\}$, let $m_i = \lvert C_i' \rvert$.\\
Let $m = m_1 + \cdots + m_k$.\\
Find (via binary search) the least index $j \in \{1, \ldots, k\}$\\
\> such that $m_1 + \cdots + m_j > m/2$.\\
If $j > 1$:\\
\> Let $\mathcal A = \textsc{Complete}(\mathcal C_1', \ldots, \mathcal C_{j - 1}')$.\\
\> Let $\mathcal B = \mathcal A \cup \{C_1' \cup \cdots \cup C_j'\} \cup \mathcal C_j'$.\\
Else:\\
\> Let $\mathcal B = \mathcal C_j'$\\
If $j < k$:\\
\> Let $\mathcal D = \textsc{Complete}(\mathcal C_{j + 1}', \ldots, \mathcal C_k')$.\\
\> Let $\mathcal C' = \mathcal B \cup \{C\} \cup \mathcal D$, where $C = C_1' \cup \cdots \cup C_k'$.\\
Else:\\
\> Let $\mathcal C' = \mathcal B$
\end{tabbing}}
\hfill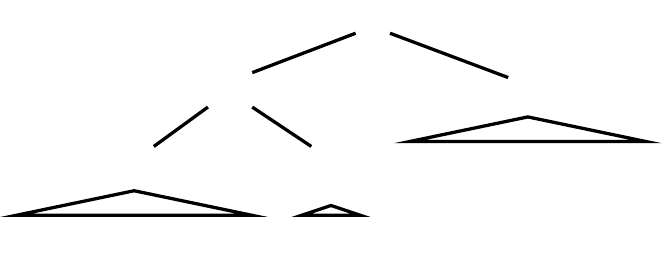
\end{framed}
\caption{Algorithm for balancing a branch decomposition.}
\end{figure*}
\begin{lemma}
\label{lem:log-height-branchwidth}
Let $\mathcal C$ be a branch decomposition rooted at $C$ of width at most $w$.
The output $\mathcal C' = \textsc{Balance}(\mathcal C)$ is a branch decomposition rooted at $C$ of width at most $2 w$, and, for all edges $e \in C$, there exist at most $3 \log_2 m + 1$ clusters $D \in \mathcal C'$ such that $e \in D$.

There exist linear-time implementations of \textsc{Balance} and \textsc{Complete}.
\end{lemma}
\begin{proof}
For all clusters $D \in \mathcal C'$, either $D \in \mathcal C$, or there exist clusters $D_1, D_2 \in \mathcal C$ with $D_1 \supsetneq D_2$ such that $D = D_1 - D_2$.
In the latter case, since $\partial D \subseteq \partial D_1 \cup \partial D_2$ and $\lvert \partial D_1 \rvert, \lvert \partial D_2 \rvert \le w$, it follows that $\lvert \partial D \rvert \le 2 w$.

In order to prove the bound on the number of clusters containing a particular edge, we define a binary function $\delta = \delta(m_1, \ldots, m_k)$.
If $m_1 = 1$ and, for all $i \in \{2, \ldots, k\}$, it holds that $m_1 + \cdots + m_{i - 1} \ge m_i$, then $\delta = 0$.
Otherwise, $\delta = 1$.
The use of a heavy path in \textsc{Balance} ensures that every root invocation of \textsc{Complete} has $\delta = 0$.

The bound follows directly from this claim, which we prove by induction on $k$: for all inputs $\mathcal C_1, \ldots, \mathcal C_k$ to \textsc{Complete}, for all $i \in \{1, \ldots, k\}$, there exist at most $3 (\log_2 m - \log_2 m_i) + 2 \delta$ clusters $D \in \mathcal C'$ such that $D \supsetneq C_i$.
The base case $k = 1$ is trivial, since $\mathcal C' = \mathcal C_1$.
When $k > 1$, let $m_{< j} = m_1 + \cdots + m_{j - 1}$ and $m_{> j} = m_{j + 1} + \cdots + m_k$ and $\delta_{< j} = \delta(m_1, \ldots, m_{j - 1})$.
By the choice of $j$, it holds that $m_{< j} \le m/2$ and $m_{> j} < m/2$.
For all $i \in \{j + 1, \ldots, k\}$, it follows by the inductive hypothesis that there exist at most $1 + 3 (\log_2 m_{> j} - \log_2 m_i) + 2 \delta_{> j} < 3 (\log_2 m - \log_2 m_i)$ clusters $D \in \mathcal C'$ such that $D \supsetneq C_i$.
For all $i \in \{1, \ldots, j - 1\}$, there exist at most $2 + 3 (\log_2 m_{< j} - \log_2 m_i) + 2 \delta_{< j} < 3 (\log_2 m - \log_2 m_i) + 2 \delta$ clusters $D \in \mathcal C'$ such that $D \supsetneq C_i$, since $2 \delta_{< j} \le 2 \delta$.
Lastly, there exist at most $2 \le 3 (\log_2 m - \log_2 m_j) + 2 \delta$ clusters $D \in \mathcal C'$ such that $D \supsetneq C_j$, since if $\delta = 0$ then $m_j \le m/2$.
\end{proof}

\subsection{A framework for Steiner forest in graphs of bounded branchwidth}

In this section, we make explicit the dynamic programming framework of Bateni
et al., adapted to branch-decomposition instead of tree-decomposition.

\begin{figure*}
\centerline{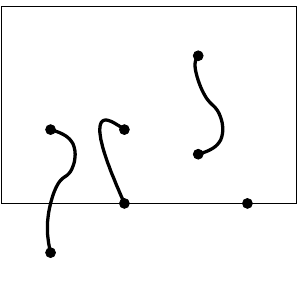}
\caption{Active vertices. Consider cluster $C$ and demand pairs $\{ s_i,t_i\}$,
$\{ s_j,t_j\}$ and $\{ s_k,t_k\}$. Here, $s_i,s_j,t_j$ and $u$ are active
vertices of $C$ but $t_i,s_k,t_k$ are not.}
\label{figure:active}
\end{figure*}

\begin{Definition}
With respect to a cluster $C$, a vertex $u\in V(C)$ is \emph{active} if it
either belongs to $\partial C$ or participates in a demand $\{u, v\}\in
\mathcal D$ such that $v \notin V(C) - \partial C$ (see
Figure~\ref{figure:active}).  We use $\act(C)$ to denote the active vertices.
A demand $\{ u, v\}\in \mathcal D$ is \emph{active} if either $u$ or $v$ is
active.
\end{Definition}

\begin{Definition}
Given two partitions $P_1$ and $P_2$, let $P_1 \vee P_2$ denote  the finest
partition coarser than both $P_1$ and $P_2$.
\end{Definition}

\begin{Definition}
With respect to a cluster $C$, a \emph{configuration} $(\piin, \piout, \piall)$
is a triple consisting of partitions $\piin, \piout$ of $\partial C$ and a
partition $\piall$ of $\act(C)$  such that $$\piall|_{\partial C} =\piin \vee
\piout .$$ For a subgraph $F$, in the \emph{canonical configuration of $C$},
$\piin$ is the connectivity of $\partial C$ in $F \cap C$, $\piout$ is the
connectivity of $\partial C$ in $F- C$, and $\piall$ is the connectivity of
$\act(C)$ in $F$ (see Figure~\ref{figure:canonicalconfiguration}).  A subgraph
$F$ and a configuration $(\piin, \piout, \piall)$ are \emph{compatible} if
$\piin$ is the connectivity of $\partial C$ in $F \cap C$ and $\piall$ is the
connectivity of $\act(C)$ in $(F \cap C) \vee \piout$ (see
Figure~\ref{figure:compatible}).  Note that compatibility is determined by the
edges in $F\cap C$ only.
\end{Definition}

\begin{figure*}
\centerline{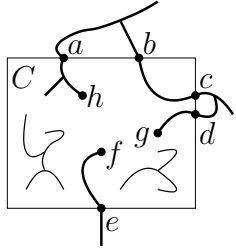}
\caption{Canonical configuration. Consider the subgraph $F$ depicted above, and
assume that $a,b,c,d,e,f,g$ are the active vertices of cluster $C$. Then the
canonical configuration of $C$ for $F$ is: $\piin = \{ a\}, \{ b,c \}, \{ d\},
\{ e \}$, $\piout =\{ a,b \},\{ c,d \}, \{ e \}$, and $\piall =\{ a,b,c,d,g,h
\},\{ e,f \}$.  }
\label{figure:canonicalconfiguration}
\end{figure*}

\begin{figure*}
\centerline{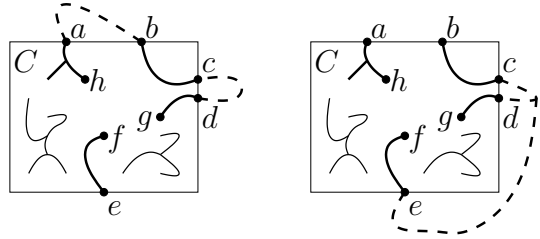}
\caption{Compatibility of a subgraph with a configuration. If the set of active
vertices is $\{ a,b,c,d,e,f,g,h\}$, then the subgraph $F$ depicted above is
compatible with the configuration $\piin =\{ a\}, \{ b,c\}, \{ d\}, \{ e\}$,
$\piout=\{ a,b\},\{ c,d \}, \{ e\}$, and $\piall=\{ a,b,c,d,g,h\},\{ e,f \}$
(left side). It is also compatible with the configuration $\sigma^{in}=\piin$,
$\sigma^{out}=\{ a\}, \{ b \}, \{ c,d,e\}$, and $\sigma^{all}=\{ a,h\}, \{
b,c,d,e,f,g\}$ (right side).  }
\label{figure:compatible}
\end{figure*}

\begin{proposition} \label{prop:canonical-config}
With respect to each cluster $C$, a subgraph $F$ is compatible with its
canonical configuration.
\end{proposition}

\begin{Definition}
Let $C_0$ be a cluster with child clusters $C_1$ and $C_2$.  For $i \in \{0, 1,
2\}$, let $(\piin_i, \piout_i, \piall_i)$ be a configuration with respect to
$C_i$.  The configurations $(\piin_0, \piout_0, \piall_0)$, $(\piin_1,
\piout_1, \piall_1)$ and $(\piin_2, \piout_2, \piall_2)$ are \emph{compatible}
if all of the following conditions hold.
\begin{itemize}
\item $\piin_0 = (\piin_1 \vee \piin_2)|_{\partial C_0}$: the internal
connectivity of the parent is the join of the internal connectivity of the
children.
\item For $i \in \{1, 2\}$, it holds that $\piout_i = (\piout_0 \vee \piin_{3 -
i})|_{\partial C_i}$: the external connectivity of a child is the join of the
external connectivity of the parent and the internal connectivity of the other
child.
\item For $i \in \{0, 1, 2\}$, it holds that $\piall_i = (\piall_0 \vee
\piall_1 \vee \piall_2)|_{\act(C_i)}$.
\end{itemize}
\end{Definition}
See Figure~\ref{figure:compatibleconfigurations} for an example.

\begin{figure*}
\centerline{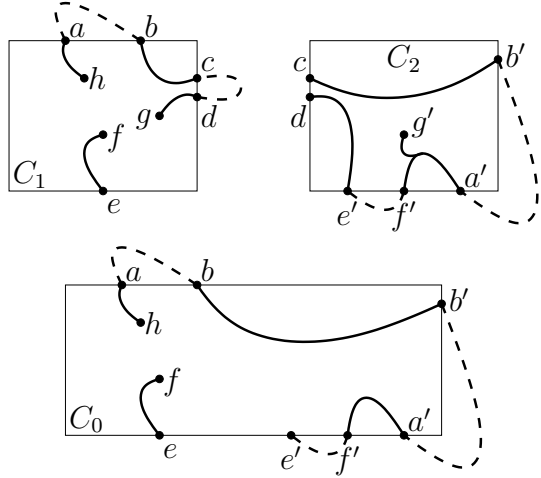}
\caption{Compatible configurations. In the above example, assume that $g$ is
active for $C_1$ and $g'$ is active for $C_2$, but neither $g$ nor $g'$ are
active for $C_0$. The dashed connections represent $\piout_i$, the solid
connections determine $\piin_i$, and the union of solid and dashed edges
determine $\piall_i$.  }
\label{figure:compatibleconfigurations}
\end{figure*}

\begin{lemma} \label{lem:dp-complete}
Let $F$ be a subgraph.  Then, for every cluster $C_0$ with child clusters $C_1$
and $C_2$, the canonical configurations of $F$ with respect to $C_0$, to $C_1$
and to $C_2$ are  compatible.

Conversely, suppose that we have a configuration for each cluster, such that
for every cluster $C_0$ with child clusters $C_1$ and $C_2$, the configurations
for $C_0$, for  $C_1$ and for $C_2$ are compatible.  Then there exists a
subgraph $F$ such that the configurations are the canonical configurations of
$F$ with respect to the clusters.
\end{lemma}
The configurations essentially give a local representation of $F$.

We now show how to use the representation of $F$ with configurations  to
determine, with local conditions,  whether $F$ is a feasible Steiner forest
solution.
\begin{Definition} \label{definition:outgoing}
With respect to a cluster $C$, a connected component is \emph{outgoing} if it
intersects $\partial C$.  A configuration is {\em outgoing}  if  every part of
$\piall$ intersects $\partial C$.
\end{Definition}

\begin{proposition} \label{prop:active-outgoing}
If $F$ is a Steiner forest solution, then, with respect to a cluster $C$, for
each active vertex $u$, the tree of $F$ containing $u$ is outgoing.
\end{proposition}

\begin{Definition} \label{definition:demand-consistent}
Let $C_0$ be a cluster with child clusters $C_1$ and $C_2$. Three
configurations for $C_0,C_1,C_2$ are \emph{demand-consistent}  if they are
compatible, outgoing, and if the following condition holds in addition: For all
demands $\{ s, t\}$ active for $C_1$ and $C_2$ but not $C_0$, terminals $s$ and
$t$ are related by $\piall_0 \vee \piall_1 \vee \piall_2$.
\end{Definition}

For example, in the example of Figure~\ref{figure:compatibleconfigurations},
the configurations are demand-consistent because terminal $g$ is related to
terminal $g'$: in $\piall_1$, $g$ is connected to $d$, and in $\piall_2$, $d$
is connected to $g'$. For an example where $\piall_0$ comes into play, see
Figure~\ref{figure:demandconsistent}.

\begin{figure*}
\centerline{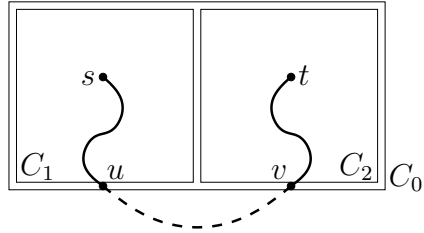}
\caption{$s$ and $t$ are related by $\piall_0 \vee \piall_1 \vee \piall_2$: $s$
is connected to $u$ via $\piall_1$, $u$ is connected to $v$ via $\piall_0$, and
$v$ is connected to $t$ via $\piall_2$.  }
\label{figure:demandconsistent}
\end{figure*}

The following lemma establishes that a subgraph represented by its canonical
configurations is a feasible solution if and only if the configurations are
consistent with one another.
\begin{lemma} \label{lem:dp-sound}
A subgraph $F$ is a Steiner tree solution if and only if, for every cluster
$C_0$ with children $C_1, C_2$,  the canonical configurations of $F$ with
respect to $C_0, C_1$ and $C_2$ are demand-consistent.
\end{lemma}

Now, suppose that for each cluster $\cal C$, we restrict attention to a subset
$\Pi_{\cal C}$ of the configurations of $\cal C$.  We call those configurations
{\em simple}.  Then the above setup leads to a dynamic program to find the
shortest Steiner forest $F$ such that for every cluster $\cal C$, the canonical
configuration of $\cal C$ for $F$ belongs to $\Pi_{\cal C}$.  The dynamic
program  works as follows.
\begin{center}
\fbox{\parbox{6in}{\vspace{-1.5ex}
\begin{tabbing}
\ \ \=\ \ \=\ \ \=\ \ \=\kill
For each cluster $C_0$ in bottom-up order,\\
\> if $C_0$ is a single edge $e$,\\
\>\> then the cost of a configuration is either $1$ or $0$\\
\>\>\> depending on whether $e$ needs to be in.\\
\> else, let $C_1$ and $C_2$ denote the two children of $C_0$;\\
\>\> for each configuration $\pi_0\in\Pi_{C_0}$,\\
\>\>\> $\text{cost} (\pi_0):= \min ( \text{cost}(\pi_1)+\text{cost}(\pi_2) )$,\\
\>\>\>\> where the min is over $\pi_1\in \Pi_{C_1}$ and $\pi_2\in \Pi_{C_2}$\\
\>\>\>\> such that $\pi_0,\pi_1,\pi_2$ are demand-consistent.
\end{tabbing}
\vspace{-1.5ex}}}
\end{center}
The runtime of that dynamic program is $O(n (\max_{\cal C} |\Pi_{\cal C}|)^3)$,
times the cost of checking that three configurations are demand-consistent.

Bateni et al.\ proposed a definition of $\Pi_{\cal C}$ such that $\max_{\cal C}
|\Pi_{\cal C}|=n^{\mbox{poly}(1/\epsilon)}$, and proved that, under their
definition, there exists a near-optimal forest such that for every $\cal C$,
the canonical configuration of $\cal C$ belongs to $\Pi_{\cal C}$. Here,
building on the spanner property, we propose a different definition of
$\Pi_{\cal C}$, such that $\max_{\cal C} |\Pi_{\cal C}|= O((\log\log
n)^{f(\epsilon)})$. This gives us the improvement from an inefficient to an
efficient approximation scheme.

In a nutshell, here is the idea: if the optimal forest has several trees that
come close to the same vertex $v$ of $\partial C$, then we connect them with
paths to transform them into a single tree; that simplifies the configuration.
We charge the length of that path to the total length of edges in the
neighborhood of $v$ in $\partial C$. To make sure that we do not charge the
same edges several times over, we contract edges that get charged.

\subsection{Contractions in a branch decomposition} \label{sec:contractions}

For a graph $G$ and a set $S$ of vertices, $SP(G, S)$ denotes the\footnote{For
uniqueness, assume the edges of $G$ are assigned distinct ID numbers, and
define $SP(G, S)$ to be the shortest-path forest in which ties are broken by ID
number.} shortest-path forest rooted at $S$. For a number $k$, define $SP(G, S,
k)=\set{e\in SP(G, S)\ : e \text{ is in a path of length at most $k$ starting
at } S}$.  That is, $SP(G, S, \eta)$ is the shortest-path forest for vertices
of $G$ whose distance from $\partial G$ is at most $\eta$.

Let $G$ be a graph with branch decomposition $\cal C$.  Fix a parameter
$\alpha$ whose value, a function of $\epsilon$ and of $w$, will be set later.
We define a recursive algorithm that operates on the clusters of $\cal C$. For
each cluster $C$ in bottom-up order, it computes a radius $\rho^C$, such that
paths of length $\rho^C$ are much shorter (by a factor of $\alpha$) than the
number of edges within distance $\rho^C$ of $\partial C$.  It then contracts
all edges within distance $\rho^C$ of $\partial C$.
\begin{center}
\fbox{\parbox{6in}{\vspace{-1ex}
\begin{tabbing}
def $\text{\sc Contract}_\alpha(C)$:\\
\quad \= $A^C = \bigcup_{C_i \text{ child of } C}  \text{\sc Contract}_\alpha(C_i)$\\
\> $\rho^C := \max \set{\rho\ :\ \len(SP(C/A^C, \partial C, \rho))
\geq \alpha \rho}$\\
\>$B^C =  SP(C/A^C, \partial C, \rho^C)\cup A^C $\\
\>return $B^C$
\end{tabbing}
\vspace{-1ex}}}
\end{center}

What are the structural properties achieved by the contraction algorithm?
First, the contracted graph $C/B^C$ has linear growth rate:
\begin{lemma} \label{lem:light-margin}
For every $\rho\geq 0$, $\len(SP(C/B^C, \rho)) \leq \alpha \rho$.
\end{lemma}
\begin{proof}
Observe the following simple property of contractions: $SP(G, S,
\eta_1+\eta_2)$ is the disjoint union of $SP(G, S, \eta_1)$ and $SP(G/SP(G, S,
\eta_1), S, \eta_2)$.  Thus, $\len(SP(C/B^C,\partial C, \rho)) =
\len(SP(C/A^C,\partial C,\rho^C+\rho)) - \len(SP(C/A^C, \partial C, \rho^C)).$
By definition of  $\rho^C$, $SP(C/A^C, \partial C, \rho^C)$ has total length at
least $ \alpha \rho^C$.  By maximality of $\rho^C$,  the length of $SP(C/A^C,
\partial C, \eta^C+\rho)$ is less than $\alpha(\eta^C+\rho)$.  The lemma
follows.
\end{proof}
Second, the sum of all radii of contracted areas is small compared to the total
length of $G$:
\begin{lemma} \label{lem:contraction-cost}
$\sum_{C \in {\cal C}} \rho^C \leq {\len(G)}/{\alpha}$.
\end{lemma}
\begin{proof}
By definition, $ \rho^C \le \len(SP(C/A^C,\partial C, \rho^C))/{\alpha}$. By
definition of $B^C$, $\len(SP(C/A^C,\partial C,
\rho^C))=\len(B^C)-\sum_{C_i\text{ child of }C}\len (B^{C_i})$. Summing over
clusters $C\in\cal C$ gives the lemma.
\end{proof}

What is the running time of the contraction algorithm? Let us explain in more
detail how to compute $\rho^C$ efficiently.  Here, for each vertex, $d[u]$
denotes the distance from $\partial C$ to $u$:
\begin{center}
\fbox{\parbox{6in}{\vspace{-2ex}
\begin{tabbing}
for $i := 1,2,3, \ldots,$\\
\quad \= $s_i := \len$\=$(\{uv \in SP(C/A^C, \partial C)\ :\ d[u] = i-1,$\\
\> \>                             $d[v]=i\})$\\
$\rho^C := \max \set{i\ :\ s_1+
\cdots + s_i \geq \alpha\, i}$
\end{tabbing}\vspace{-2ex}}}\\
\end{center}
The total time for each invocation $\text{\sc Contract}_\alpha(C)$ is linear in
the number of edges of $C$.  By Lemma~\ref{lem:log-height-branchwidth}, ${\cal
C}$ is a log-height branch-decomposition, so the total time for calling
$\text{\sc Contract}_\alpha$ on the root cluster of a graph $G$ of $O(n)$ edges
is $O(n \log n)$.

\subsection{Regions covering partially contracted clusters} \label{sec:regions}

Here is a high-level description of our method for finding regions for a
cluster $C$.  Fix a parameter $\beta$ whose value, a function of $\epsilon$ and
of $w$, will be set later.

\begin{center}
\fbox{\parbox{6in}{\vspace{-1.5ex}\begin{tabbing}
\quad\=\quad\=\quad\=\quad\=\kill
Denote by $2^{\mu^C}$ the minimum power of two\\
\> greater than $\max_{u \in \act(C)} \dist(u, \partial C)$.\\
For each $i=1,2,\cdots, \mu^C$,\\
\> define a set $\mathcal L_i$ of {\em regions} of $C/B^C$ such that\\
\>\> each region has diameter at most $\beta 2^i$, and\\
\>\> together, the regions of $\mathcal L_i$ cover $SP(C/B^C, 2^i)$.
\end{tabbing}\vspace{-1.5ex}}}
\end{center}

It follows from Lemma~\ref{lem:light-margin} that a greedy algorithm produces
such a covering, of size $|{\mathcal L_i}|=O(\alpha/\beta + \lvert \partial C
\rvert)$.  However, in order to get near-linear running time, we need an
algorithm that is slightly more sophisticated than greedy. In the rest of this
subsection, we present the details of the algorithm.

For each tree of $SP(C/B^C, \mu^C)$, let $v$ be a tree vertex that is on the
boundary $\partial C$ and let $\mathcal T_v$ denote the sequence of vertices
encountered on an Euler tour of  the tree.  An {\em $i$-region} is a subpath of
$\mathcal T_v^C$ of length at most $\beta 2^i$ whose first vertex is at a
distance at most $2^i+\beta 2^i$ from $\partial C$ in $C/B^C$.  In our
algorithm,  $\mathcal L_i$ is a set of $i$-regions.  For a vertex $u$ of $C$,
we use $d[u]$ to denote the $\partial C$-to-$u$ distance in $C/B^C$.

\begin{center}
\fbox{
\parbox{6in}{
\vspace{-1.5ex}
\begin{tabbing}
\quad\=\quad\=\quad\=\quad\=\quad\=\quad\=\quad\=\quad\=\quad\=\kill
for $i :=1,2,\ldots , \mu^C$, \\
\> $\mathcal L_i  := \emptyset $\\
for each tree of $SP(C/B^C, \mu^C)$,\\
\> root the tree at some vertex $v\in \partial C$\\
\> construct an Euler tour $T_v$ of the tree\\
\> for $i := \mu^C,\mu^C-1,\ldots, 1$,\\
\>\> $S_{v, i} := \{  j \beta 2^i\ :\    j \in \set{0,1,2,3,\ldots}$\\
\>\>\>\>\>\>\>\>\> $\text{and }  d[T_v[j\beta 2^i] ] \leq (1+\beta) 2^i   \}$\\
\>\> for each $x\in S_{v,i}$,\\
\>\>\> ${\mathcal L_i} :={\mathcal L_i}\cup{}$\\
\>\>\>\> $\{  \text{subpath of $T_v$ of length $\beta 2^i$}$\\
\>\>\>\> $\text{truncated at $|T_v|$) starting at   $T_v[x]$} \}$
\end{tabbing}
\vspace{-1.5ex}
} }
\end{center}

\begin{lemma}
For each $i=1,2,\cdots , {\mu^C}$, the regions of $\mathcal L_i$ cover
$SP(C/B^C, 2^i)$.
\end{lemma}

\begin{lemma} \label{lem:few-regions}
For each  $i=1,2,\cdots , {\mu^C}$, $\lvert \mathcal L_i \rvert \leq 2
\alpha(1+2\beta)\beta^{-1} + \lvert \partial C \rvert$.
\end{lemma}
\begin{proof}
If a vertex $u$ is in a subpath of  $\mathcal L_i$, then it is at distance at
most $\beta 2^i$ from the starting point of the subpath, and so $u\in SP(C/B^C,
(1+2\beta)2^i)$.  So, the sum of lengths of the $i$-regions is at most the
length of the Euler tours of $SP(C/B^C, (1+2\beta)2^i)$, which is at most twice
the number of edges in $SP(C/B^C, (1+2\beta)2^i)$.  By
Lemma~\ref{lem:light-margin},  $|SP(C/B^C, (1+2\beta)2^i)|$ is at most $\alpha
(1+2\beta )2^i$.  Each $i$-region constructed from $T_v$ has length exactly
$\beta 2^i$,  except possibly the last one.   The lemma follows.
\end{proof}
The time for finding the covers is $O(n \log n)$.

\subsection{Simple configurations}

\begin{Definition}
Fix a cluster $C$ with boundary vertices $\partial C$ and, for each
$i=1,2,\cdots ,{\mu^C}$, a covering $\mathcal L_i$ of $C/B^C$ by $i$-regions.
Given
\begin{itemize}
\item an integer $d\leq |\partial C|$;
\item $d$ powers of $2$ in the range $[2^{\mu^C}/(2\gamma),2^{\mu^C}]$, where
$\gamma$ is a parameter to be determined later;
\item a priority ordering over those $d$ numbers, labelled
$(2^{i_1},2^{i_2},\ldots ,2^{i_d})$ by order of priority; and
\item for each $i_j$, a set of $i_j$-regions $\mathcal Q_j \subseteq \mathcal
L_{i_j}$;
\end{itemize}
consider the subpartition of $\act(C)$, denoted  $P(i_1, \ldots, i_d, \mathcal
Q_1, \ldots, \mathcal Q_d)$ , and defined by greedily setting the $j$th part to
be $$P_j = \bigl(\act(C) \cap \text{Uncontract}(\bigcup_{Q \in \mathcal Q_j}
Q\bigr)) - \bigcup_{\ell = 1}^{j - 1} P_\ell ,$$ where $\text{Uncontract}$
takes as input vertices of $C/B^C$ and outputs the corresponding vertices of
$C$.  A configuration $(\piin, \piout, \piall)$ is {\em simple} iff $\piall =
\piin \vee \piout \vee P(i_1, \ldots, i_d, \mathcal Q_1, \ldots, \mathcal Q_d)$
for some $(i_1, \ldots, i_d, \mathcal Q_1, \ldots, \mathcal Q_d)$.
\end{Definition}

To understand this definition intuitively, $d$ should be interpreted as the
number of outgoing trees. The $d$ powers of $2$ should be interpreted as the
approximate ``radii" of those trees -- maximum distance from $\partial C$ to an
active tree vertex. As in the algorithm of Bateni et al., the ordering should
be interpreted as giving priority to trees whose minimal enclosing cluster is
smaller. As in the algorithm of Bateni et al., the (uncontracted) $i_j$-regions
should be interpreted as a covering of the active vertices of the $j$th tree.

\begin{center}
\fbox{\parbox{6in}{\vspace{-1.5ex}
\begin{tabbing}
For each cluster $C$\\
\quad \= define a table indexed by $(i_1, \ldots, i_d, \mathcal Q_1, \ldots, \mathcal Q_d)$
\end{tabbing}
\vspace{-1.5ex}}}
\end{center}

\begin{lemma} \label{lem:number-gamma-simple-configs}
Taking $\alpha$ and $\beta$ to be constant, the number of simple configurations
is $(\log_2 \gamma)^{f(\epsilon)}$ for some function $f$ of $\epsilon$.  The
time to check demand-consistency is $O(\log n)$.
\end{lemma}

We can finally state the main structural Theorem that is at the core of
Theorem~\ref{thm:fast-dynamic-programming}.

\begin{theorem} \label{thm:DP-approximation}
For any solution $F$,  there exists a solution $F' \supseteq F$ such that for
every cluster $C$, the canonical configuration of $F'$ with respect to $C$ is
simple, and whose length satisfies: $\len(F') \le \len(F) + 4 \beta (2w - 1)
\bigl(1 + (3 \log_2 m + 1)/\gamma\bigr) \len(F) + 2 \alpha^{-1} (2w - 1) \len(G)
.$
\end{theorem}

\subsection{Proof of Theorem~\ref{thm:DP-approximation}}

\subsubsection{Defining $F'$}

$F'$ is simply an extension obtained from $F$ by adding some edges.  First,
given $F$ and a cluster $C$, we define a partition of the form $P(i_1, \ldots,
i_d, \mathcal Q_1, \ldots, \mathcal Q_d)$.

Let $d$ be the number of outgoing trees of $F$.  Label these trees $T_1,
\ldots, T_d$ in a particular order, such that the following property holds: if
the minimal cluster including $E(T_j)$ is a proper descendant of the minimal
cluster including $E(T_\ell)$, then $j < \ell$.  We now define $i_j'  = \lceil
\log_2 \max_{u \in V(T_j) \cap \act(C)} \dist(u, \partial C) \rceil $ and $ i_j
= \max\{i_j',\:\mu - \lceil \log_2 \gamma \rceil - 1\}$ and let $\mathcal Q_j
\subseteq \mathcal L_j$ be a minimal set of regions such that $\bigcup_{Q \in
\mathcal Q_j} Q$ covers the vertices of $(T_j \cap C)/B^C$. This defines the
radii and sets of regions, hence also specifies the associated partition
$(P_1,\ldots ,P_d)$ of $\act(C)$.

The construction of $F'$ is in two steps. First we go from $F$ to $F_1$ by
adding some edges, then we go from $F_1$ to $F'$ by adding more edges.

First step: starting from $F_1:= F$, modify $F_1$ by processing clusters $C \in
\mathcal C$ in top-down order. Consider a cluster $C$.  While there exists a
$j$ and an active vertex $u$ that is in $P_j$ but is not connected to $T_j/B^C$
in $F_1/B^C$, add to $F_1$ a path of $C/B^C$ connecting $u$ to $T_j/B^C$.
(This first step is similar to the construction in Bateni el al.)

Second step: starting from $F':= F_1$, modify $F'$ by processing clusters $C
\in \mathcal C$ in top-down order. Consider a cluster $C$.  While there exists
a pair of boundary vertices $u, v \in \partial C$ such that  $u$ and $v$ are
not connected in $(F'\cap C)/A^C$ but can be connected by adding at most
$2\rho^C$ edges of $C/A^C$ to $F'$, add those edges to $F'$.

The result of this processing defines $F'$.

\subsubsection{The canonical configuration of $F'$ is simple}

\begin{lemma} \label{lemma:firststep}
After the first step of the construction, for every cluster $C$ we have: all
active vertices of $P_j$ are connected to $T_j$ in $(F_1\cap C)/B^C$.
\end{lemma}

\begin{lemma} \label{lemma:secondstep}
After the second step of the construction, for every cluster $C$ we have: if
two vertices of $\partial C$ are connected in $(F'\cap C)/B^C$ then they are
connected in $F'\cap C$.
\end{lemma}
\begin{proof}
It suffices to show that, after those paths are added, if two vertices of
$\partial C$ are connected in $(F'\cap C)/B^C$ then they are connected in
$(F'\cap C)/A^C$.  Suppose not, and let $p$ be a $\partial C$-to-$\partial C$
path in $(F'\cap C)/B^C$ that is not a path in $(F'\cap C)/A^C$, chosen so as
to have a minimal number of edges.  Then, in $C/A^C$,  $p$ starts at some
vertex $p_{\text{start}}$ of $SP(C/A^C,u,\rho^C)$ and ends at some vertex
$p_{\text{end}}$ of $SP(C/A^C,v,\rho^C)$. Concatenating $p$ with a path from
$u$ to $p_{\text{start}}$ at one end, and with a path from $p_{\text{end}}$ to
$v$ at the other end, gives a path from $u$ to $v$ in $C/A^C$. The total number
of edges thus added is at most $2\rho^C$, so it would have been added during
the processing.
\end{proof}

\begin{lemma}
For every cluster $C$, the canonical configuration of $F'$ is simple.
\end{lemma}
\begin{proof}
It suffices to show for all $j$ that all active vertices in $P_j$ are connected
by $F'$ to $T_j$.  Let $u'$ be an active vertex in $P_j$.  Either in $F$ $u'$
belongs to $T_j$, in which case there is nothing to show, or $u'$ belongs to
some tree $T_\ell$ of $F$. Thus $u'$ is covered by some region of  $\mathcal
Q_{i_j}$ and by some region of $\mathcal Q_{i_\ell}$. The fact that $u'\in P_j$
indicates, by definition of $P_j$, that $j$ must be less than $\ell$.  Then, by
the first step of the construction (Lemma~\ref{lemma:firststep}), $u'$ get
connected to $T_j$ in $(F_1\cap C)/B^C$.  Since $T_\ell$ and $T_j$ are both
outgoing, in $F$ $u'$ is connected to some vertex $u\in \partial C$ and $T_j$
is connected to some vertex $v\in\partial C$.  By transitivity  $u$ is
connected to $v$ in $(F_1\cap C)/B^C$.  By the second step of the construction
(Lemma~\ref{lemma:secondstep}), $u$ is connected to $v$ in $F'\cap C$.  By
transitivity again, $u'$ is connected to $T_j$ in $F'\cap C$.
\end{proof}

\subsubsection{Length of $F'$}

First we analyze the length increase when gong from $F$ to $F_1$.  Since $u\in
P_j$, there exists a region $Q \in \mathcal Q_j$ such that $u \in Q$.  By
definition, $Q$ covers at least one  vertex $w \in T_j$.  Thus the path added
to connect $u$ to $T_j$ in $(F'\cap C)/B^C$ has length at most $2\beta
2^{i_j}$.  If $i_j = i_j'$, then the length is at most $4 \beta \len(T_j \cap
C)$.    In this case, we charge the length of this path to $T_j$.  Otherwise,
it's at most $4 \beta \len(F \cap C)/\gamma$.  In this case, we charge the
length of this path to $F\cap C$.

We claim that each tree $T_j$ of $F$ is charged at most $2w - 1$ times by paths
added when $i_j = i_j'$.  Indeed, whenever $T_j$ is charged, some other tree
$T_\ell$ of $F$ is connected to $T_j$ in $(F'\cap C)/B^C$.  After processing
descendant clusters of $C$ and adding connections between boundary vertices,
$T_j$ is connected to $T_\ell$ in $F'\cap C$.  Since $j<\ell$, the minimal
cluster $C_j$ strictly enclosing $T_j$ is either the same as for $T_\ell$ or
one of its descendant clusters.  Either way, $T_\ell$ must contain a vertex of
$\partial C_j$, so there are only at most $2w-1$ such trees $T_\ell$, so $T_j$
is charged at most $2w-1$ times.  Summing over trees $T_j$, the total charge of
those paths is at most $4 \beta (2w - 1) \len(F)$.

The length charged to $F\cap C$ is at most $4 \beta (2w - 1) (3 \log_2 m +
1)/\gamma$, since each edge is charged at most $(2w - 1) (3 \log_2 m + 1)$
times.

Second, we claim in the second step, going from $F_1$ to $F'$, each cluster $C$
is charged at most $2w-1$ times. Indeed, each charge corresponds to two boundary
vertices of $C$ being connected by a path, and after $2w-1$ paths are added, all
of $\partial C$ must be connected. Summing over clusters $C$, the total charge
of those paths is at most $\sum_{C \in \mathcal C} 2 (2w- 1) \mathbf \rho^C$,
which is at most $2 (2w- 1) \len(G)/\alpha$ by Lemma~\ref{lem:contraction-cost}.

This completes the proof of Theorem~\ref{thm:DP-approximation}.

\subsection{Proof of Theorem~\ref{thm:fast-dynamic-programming}}

Let $\gamma=(3\log_2 m+1)$, $\beta= \epsilon/8(2w-1)$, and
$\alpha=2\epsilon^{-1} c (2w-1)$ where $c$ is the constant specified in
Theorem~\ref{thm:fast-dynamic-programming}.
Lemma~\ref{lem:number-gamma-simple-configs} implies that the DP takes time $n
\log n (\log \log n)^{O(1)}$, which is $O(n \log^2 n)$.  By
Lemma~\ref{thm:DP-approximation}, the quality of the output satisfies
Theorem~\ref{thm:fast-dynamic-programming}.

\bibliographystyle{plain}
\bibliography{all}

\end{document}